\documentclass[notitlepage,11pt]{article}
\usepackage{amsfonts,amsmath,amssymb,mathrsfs,url,authblk, %MnSymbol,
geometry,esint,subfig,
comment%, cases
}
\usepackage[normalem]{ulem}
\usepackage{hyperref}
\usepackage{cleveref}

\usepackage{color}

\usepackage{graphicx}
\usepackage{epsfig}
\usepackage{float,multirow}
\usepackage[novbox]{pdfsync}

\usepackage{wasysym}
\usepackage{empheq}
\usepackage{url}

\newcommand{\RN}[1]{\textup{\uppercase\expandafter{\romannumeral#1}}}

\newcommand{\al}{\alpha}

\newcommand{\Th}{\mbox{$\mathcal{T}_h$}}

\newcommand{\R}{\mathbb{R}}

\newcommand{\h}{\; \mathsf h}
\newcommand{\q}{\; \mathsf q}

\def\bep{\begin{pmatrix}} \def\eep{\end{pmatrix}}

\def\c{\; \mathsf c}

\providecommand{\pp}[1]{\left[#1\right]} %[.]
\providecommand{\pr}[1]{\left(#1\right)} %(.)

\newtheorem{theorem}{Theorem}[section]
\newtheorem{lemma}[theorem]{Lemma}
\newtheorem{proposition}[theorem]{Proposition}

\newtheorem{remark}[theorem]{Remark}
\newtheorem{definition}[theorem]{Definition}

\renewcommand{\theta}{\vartheta}

\renewcommand{\thefootnote}{\fnsymbol{footnote}}

\numberwithin{equation}{section}

\def\bc{\begin{cases}
  }     
\def\ec{\end{cases}}

  \newcommand{\beq}{\begin{eqnarray}
    }
\def\eeq{\end{eqnarray}}
   \newcommand{\be}[1]{\begin{equation}\label{#1}}
\newcommand{\ee}{\end{equation}}
\def\bea{\begin{eqnarray*}}\def\ssec{\subsection}
    
\def\eea{\end{eqnarray*}} \def\la{\label}\def\fe{for example }   \def\ith{it holds that } 
    \def\sats{satisfies}  \def\saty{satisfy}     \def\thr{therefore}  

\def\I{\infty} 
  \def\T{\widetilde}
\def\PH{phase-type }  
\def\BEN{\begin{enumerate}}  \def\BI{\begin{itemize}}
\def\EEN{\end{enumerate}}   \def\EI{\end{itemize}} \def\im{\item} \def\Lra{\Longrightarrow}  \def\eqr{\eqref}  
 \newcommand{\dom}{{\mathcal D}}\newcommand{\frd}[2]{\frac{\text{d} #1}{\text{d} #2}}
\def\mR{\mathcal R} \def\oth{otherwise}
\def\g{\gamma}     \def\b{\beta}

 \def\lam{\lambda}  \def\al{\alpha}
\def\Fr{Furthermore, }
   \def\mbe{may be expressed in terms of } 
    
   \def\wrt{with respect to }
  \def\resp{respectively} \def\how{however} \def\fno{from now on} 
 \def\eqr{\eqref}  
\def\wk{well-known}

\def\fr{\frac} \def\im{\item}
\newcommand{\s}{\;\mathsf s}
\newcommand{\e}{\;\mathsf e}
\renewcommand{\i}{\;\mathsf i}
\renewcommand{\r}{\; \mathsf r}
\renewcommand{\v}{\; \mathsf v}
\def\eeD{\end{definition}} \def\beD{\begin{definition}}
\def\beR{\begin{remark}} \def\eeR{\end{remark}}
\def\beL{\begin{lemma}} \def\eeL{\end{lemma}}\newcommand{\f}[2]{\frac{#1}{#2}}

      \def\muR{\mu_{\mR}}\def\bb{\bff \b}  \def\ba{\; \bff a} \def\vi{\vec \i \;}  
    \def\va{\vec \alpha } \def\vJ{\vec J }
     \def\vr{\; \vec {\mathsf r}}\newcommand{\bff}[1]{{\mbox{\boldmath$#1$}}}
    \def\beP{\begin{proposition}} \def\eeP{\end{proposition}}
    \newcommand{\vI}[1]{\vec I^{(#1)}}  
    \def\Th{\Theta}
    
    \long\def\symbolfootnote[#1]#2{
\begingroup
\def\thefootnote{\fnsymbol{footnote}}\footnote[#1]{#2}
\endgroup}
\def\fn{\symbolfootnote}
\def\Thr{Therefore, }\def\det{deterministic }\def\corr{corresponding }\def\Fno{From now on, }
\newcommand{\lt}{\left}\newcommand{\rt}{\right}
\newtheorem{example}{Example}\def\beXa{\begin{example}} \def\eeXa{\end{example}}

\def\a{\; \mathsf a} \def\y{\; \mathsf y}  \def\frt{furthermore }

\newcommand{\figu}[3]{
\begin{figure}[H]
\centering
\includegraphics[scale=#3]{#1}
\caption{#2\label{f:#1}}
\end{figure}
}
\renewcommand{\(}{\left(}
\renewcommand{\)}{\right)}

\newcommand{\firstargument}{}
\newcommand{\address}[2][]{\renewcommand{\firstargument}{#1}\gdef\@address{#2}}

\date{\today}

\title {%Aproximate SIR control of the matrix SIR models with ICU constraint
A review of matrix SIR Arino  epidemic models %reproduction number  $\mR$ and of the  -Feng invariant  for
}

 \author[1]{Florin Avram \footnote{Corresponding author. E-mail:
     Florin.Avram@univ-Pau.fr}}
 \author[2]{Rim Adenane}
 \author[3]{David I. Ketcheson}
 \affil[1]{Laboratoire de Math\'{e}matiques Appliqu\'{e}es, Universit\'{e} de Pau, France 64000 }
 \affil[2]{Département des Mathématiques, Université Ibn-Tofail, Kenitra, Maroc 14000}
 \affil[3]{King Abdullah University of Science and Technology, Thuwal 23955, Arabie saoudite}

\begin{document}

\maketitle

\begin{abstract}

Many of the models used nowadays in mathematical epidemiology, in particular in COVID-19 research, belong  to a certain sub-class of compartmental  models whose classes  may be divided into three ``$(x,y,z)$"  groups, which we will call  \resp ``susceptible/entrance, diseased, and output" (in the classic SIR case, there is only one class of each type). Roughly, the ODE dynamics of these models contains only linear terms, with the exception of products between $x$ and $y$ terms. It has long been noticed that   the basic reproduction number  $\mR$ has a very simple formula \eqr{R} in terms of the matrices which define the model, and an explicit  first integral  formula \eqr{Y} is also  available. These results can be traced back at least to \cite{Arino} and \cite{Feng}, \resp, and may be viewed as the ``basic laws of SIR-type epidemics"; \how\ many  papers  continue to reprove them  in particular instances (by the next-generation matrix method or by direct computations, which are unnecessary). This motivated us to redraw the attention to these basic laws and provide a self-contained reference of  related formulas for $(x,y,z)$ models.  {We propose to rebaptize the class to which they apply as matrix SIR   epidemic models, abbreviated as SYR, to emphasize  the similarity to the classic SIR case.  For the case of one susceptible class we propose  to use the name SIR-PH, due to a simple probabilistic interpretation  as  SIR models where the exponential infection time has been replaced by a PH-type distribution.}
We note  that to each SIR-PH  model,  one may associate   a  scalar quantity $Y(t)$  which \sats ``classic  SIR   relations" -- see \eqr{Y}. In  the case of several  susceptible classes this generalizes to \eqr{Yt}; in a future paper, we will show that \eqr{Y}, \eqr{Yt}    may be used to obtain approximate control policies which
compare well with the optimal control of the original model. %, in several examples.

\end{abstract}
%\keyword{epidemiological modelling; COVID-19; SIR-PH model; matrix SIR model; basic reproduction number; first integral.}
\tableofcontents
\section{Introduction}
{\bf Motivation}.
{Mathematical epidemiology  may be said to have started with the SIR ODE model, which saw its birth in the work of
Kermack–McKendrick  \cite{KeMcK}. This was initially applied to model the  Bombay plague of 1905-06, and later to  measles \cite{earn2008light}, smallpox,
chickenpox, mumps, typhoid fever and diphtheria, and recently to the COVID-19 pandemic -- see \fe\ \cite{Schaback,bacaer2020modele,Ketch,Charp,Djidjou,Sofonea,alvarez2020simple,horstmeyer2020balancing,
di2020optimal,Franco,baker2020reactive,caulkins2020long,caulkins2021optimal}, to cite just a few representatives of a huge literature.}

 Note that during the COVID-19 pandemic, researchers have relied mostly on models with  quadratic interactions (linear force of infection), which belong \frt\ to a particular class \cite{Arino,Andr,Riano,Fre20} of ``$(x,y,z)$" models.  Here $x$ denotes ``entrance/susceptible" classes, $y$ denotes diseased classes, which must converge asymptotically to $0$,  and $z$ denotes output classes. These models are very useful; to make references to them easier, we propose to call them matrix-SIR (SYR)  models, and  also SIR-PH \cite{Riano}, when $x \in \R^1$.

{\bf Contents}. We begin by recalling in Section \ref{s:SIR} several  basic explicit formulas  for the SIR model.  Section \ref{s:Feng} presents the corresponding SIR-PH generalizations, and Section \ref {s:exa} offers some applications: the SEIHRD model
\cite{ivorra2017stability,Palmer,pazos2020control,nave2020theta,ramos2021simple} which adds to the classic SEIR (susceptible+exposed+infectious+recovered)  a  hospitalized (H) class  and a dead class (D), the SEICHRD model \cite{Kantner}
which adds a critically ill class (C),
the SEIARD  \cite{de2020data} and SEIAHR/SEIRAH(D) models \cite{deng2021extended,otoo2021estimating,wang2020evolving,kucharski2020early,hayhoe2020data,khatua2020fuzzy,prague2020population}, which add an asymptomatic class (A), and the S$I^{3}$QR model \cite{shaw2021reproductive}.
This is just a sample chosen from some of our favorite COVID papers. We note in passing that they seem though all unaware of the existence of the Arino and Feng formulas \eqr{R},  \eqr{Y}. Like most  papers nowadays, they do not recognize the matrix-SIR particular case, and    continue to reprove it   (by the Jacobian, next-generation matrix, ot Chavez-Feng-Huang methods for $\mR$ \cite{Mart}, or by direct computations for the Feng formula), which have become  superfluous once the matrix-SIR particular case is recognized. We also note in passing that the concept of epidemic still seems to lack a  mathematical definition. A definition of the most common particular case is offered in \cite{Breda}; this framework is  more general  than matrix SIR by allowing age-dependence, but the Feng invariant is not discussed  there.

Finally, Section \ref {s:het}  reviews briefly the
case of several classes of susceptibles. This topic requires further development; we include it however due to the recognized importance of heterogeneity factors. % (like  the average contact rates between different age or geographical groups).

\section{The classic Kermack–McKendrick SIR epidemic  model \la{s:SIR}}

The SIR process $(S(t),I(t), R(t), t \geq 0)$ divides a population of size $N$ undergoing an epidemic into three classes called ``susceptibles, infectives and  removed". One may also work  with the \corr fractions $\s(t)=\fr{S(t)}N,\i(t)=\fr{I(t)}N, $
and  $\r(t)=1-\s(t)-\i(t)$.
 It is assumed that only susceptible
individuals can get infected.  After having been infectious for
some time, an
individual recovers and may not become susceptible again. ``Viewed from far away", this yields the  SIR model with demography  \cite{kermack1927contribution,Chavez}
\begin{align}
S'(t)&=-\fr{\b}N S (t) I (t)+ \xi \pr{N  - S(t)},\nonumber\\ %+\rho R(t)
I'(t)&= I(t)\pr{\fr{\b}N  S(t)-\g -\xi},\la{SIR}
\\
R'(t)&= \g  I(t)-\xi R(t),\nonumber %-\rho R(t)
\end{align}
where
\BEN \im N is the total, constant population size. \im $R'(t)$, the number of removed per unit time, is the only quantity which is clearly observable, at least in the easy case when the removed are dead, as was the case of the original study of the Bombay plague \cite{kermack1927contribution}. \im $\xi$ is the population death rate, assumed to equal the birth rate.
\im $\g$ is the  removal rate of the infectious, which equals 1/duration of the infection (under the stochastic model of exponential infection durations, this is the reciprocal of the expected duration).
\im $\b$, the  infection rate, models the probability  that a contact takes place between an infected and a susceptible, and that it results in  infection.
 \EEN

Note that  \BEN
  \im The sum $S+I+R =N$ is conserved and each value is positive, so the values of $S, I, R$ remain in the interval $[0,N]$.
  \im  This system has a unique  solution, since (given the boundedness of $S, I$, and $R$), the RHS above is Lipschitz.
\EEN

\Fno we will assume that $\xi=0$, and work with the fractions $\s,\i,\r$, which \saty
  \begin{align}
\s'(t)&=  -\b \s(t)\i(t) ,\nonumber\\
\i'(t)&= \i(t)\pp{\b \s(t)-\g},\label{sir}
\\
\r'(t)&= \g \i(t).\nonumber
\end{align}

We will call this the classic SIR model.
Note that
  \BEN \im   $\s(t)$ is monotonically decreasing and $ \r(t)$ is monotonically  increasing, to,  say, $s_\I, r_\I$;   therefore  convergence to some fixed stable point $(s_\I, i_\I, r_\I)$ must hold.

  \im
   the equilibrium set of stable points  is
$(s,0,1-s),      s \in [0,1].$

\im solutions starting in the domain
$$\dom:=\left\{(s,i,r): s>0,i>0,r \geq 0, s+i+r \le 1\right\}$$
    cannot leave it.

\im
The second equation of \eqr{SIR} implies  the so-called {\bf threshold phenomenon}:
if
\be{RS} \mR :=\fr{\b}{\g} \leq 1\ee then $\i(t)$ decreases always, without any intervention.

To  avoid trivialities,  we will assume $\mR >1$ \fno.

\im When $\mR >1$, the    epidemic grows iff $\s > 1/\mR$, i.e. until  the susceptibles $\s(t)$ reach the {\bf immunity threshold}
\be{IT} \Th:=\fr 1 {\mR}=\fr{\g}{\b},\ee
after which  infections decline. $\mR$ is called {\bf basic reproduction number}, and it models the number of susceptibles infected by one infectious (expected number, under more sophisticated stochastic, branching models).
\EEN

An advantage of the classic SIR model is that it is essentially solvable explicitly:
\BEN
\im
We can eliminate $\r$ from the system using the invariant $\s + \i + \r = 1$ (this is also possible for various generalizations like SIR with demography, as long as $\r$ does not appear explicitly in the rest of the equations).

\im

It can easily be verified that \be{inv}\mu(s,i):=s+i-\fr 1 {\mR}\ln(s)\ee
is  invariant, so that
 $i$ is explicitly given by
\be{iR} i_{\mR}(s) =-s +\fr 1 {\mR}\ln(s) + \muR(s_0,i_0),\ee
and the full system \eqref{sir} can be reduced to the single ODE
\be{1ode}
    \s'(t) = -\beta \s(s_0 + i_0 - \s) - \gamma \s \ln\left(\frac{\s}{s_0}\right).
\ee

\im

The  maximal value of the infected $\i$, achieved when $\s=1/\mR$,  is
 \be{im}i_{\max}=i_{\max,\mR}(s_0,i_0)= i_0+s_0 -  \fr {1+ \ln(s_0 \mR)} {\mR}  =i_0 + \f{ H(s_0 \mR)}{\mR}, \; H(\mR)= {\mR-1-\ln(\mR)}.
\ee
%where we put \be{h} H(\mR)= {\mR-1-\ln(\mR)}.\ee

\im By differentiating the right-hand side of \eqref{1ode}, one finds that the
maximal value of the ``newly infected" $-\s'=\b \i \s$ is achieved when
\be{ni} \s =\i +\fr \g \b, \s=-\fr 1{2 \mR} L_{-1}\pp{-2 \mR s_0 e^{-1-\mR (s_0+i_0)}},\ee
where  the Lambert function $L_{-1}$ is a real inverse of $L(z)=z e^z$-- see \fe
\cite{pakes2015lambert,kroger2020analytical,berberan2020exact}.
Bounding $ \s \i$ is one interesting possibility for accomodating  ICU constraints \cite[(2.20)]{Mangat}.

\im The infectious class
converges to $0$ and the susceptible and recovered converge monotonically to limits which may be expressed in terms of the ``Lambert-W(right)" function $L_0$.
\EEN

Let us note that accurate numerical solutions of  the evolution of the SIR or other compartmental epidemic may be obtained very quickly.
\figu{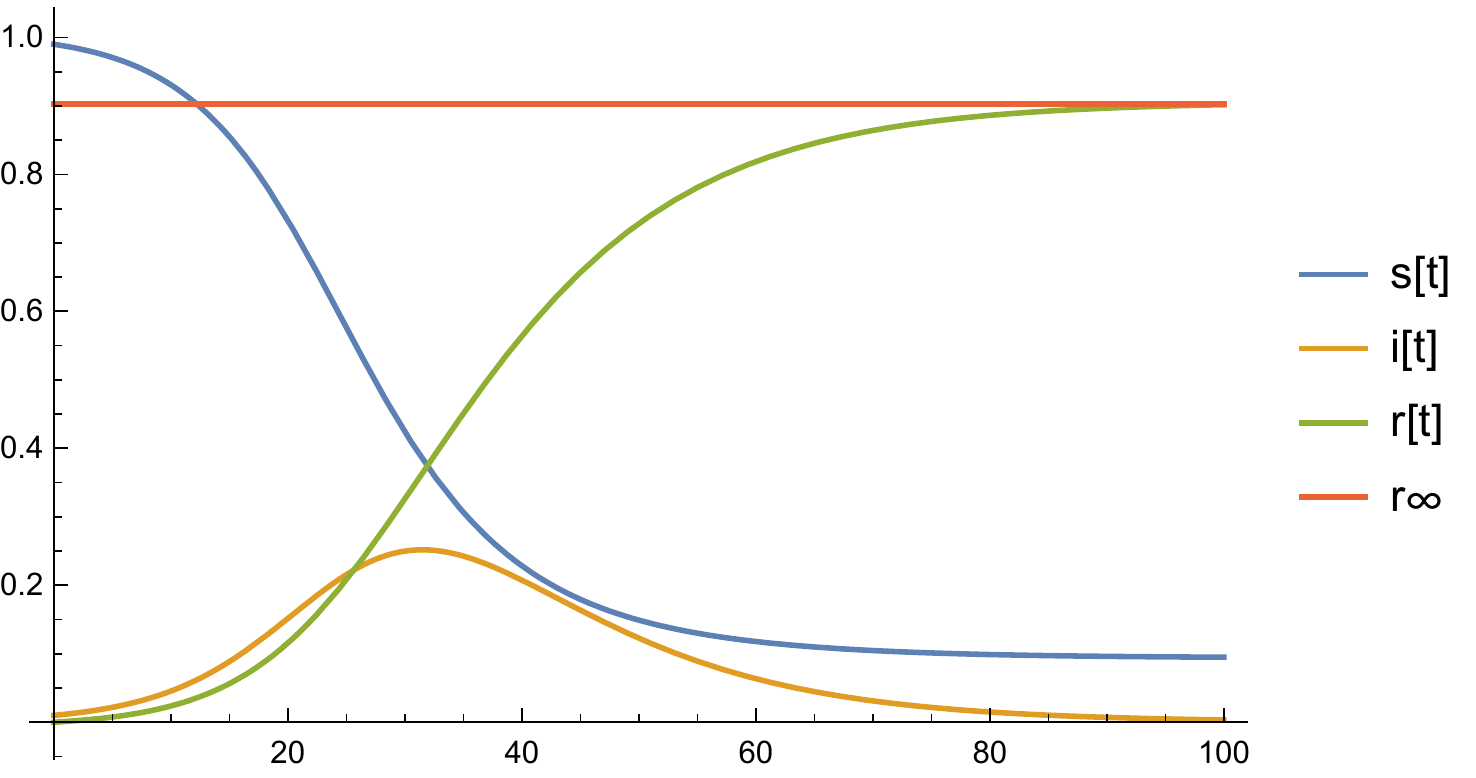}{Plot of the states of \eqr{sir}, for $\mR=2.574,\; \g=.1,\; s_0=.99,\; i_0=1-s_0,\; r_0=0 \Lra r_{\I} =s_0+i_0+ \frac{L_0[-\mR s_0 e^{-\mR(s_0+i_0)}]}{\mR}=0.903171$.
    %,  d_{\I}= \frac{\g_d A}{\g_r+\g_d}= 0.37355 $
   %(see Section \ref{s:back} below
   }{.8}%SIRD

\section{SIR-PH epidemics with one susceptible class  (SIR epidemics with \PH ``disease time"))
\la{s:Feng}}

 It has been known for  a long while that $\mR$ and  the final size %,  and approximations of  the highest peak of
 for many compartmental model epidemics may be   explicitly expressed in terms of the matrices which define  the  model, and \cite{ma2006generality,Arino,Feng,Andr} offer  a   quite   general framework of ``xyz" models which ensures this. We believe that these  formulas have not received the attention they deserve (they keep being reproved), and decided therefore to review them  below; we will call them matrix- SIR models.

  A particular but revealing case
is that when there is only one susceptible class, which we will call  SIR-PH, following Riano \cite{Riano}, who emphasized its probabilistic interpretation -- see also \cite{Hurtado}.
\beD  A ``SIR-PH $(\va,V,\bb,W)$ epidemic" contains  a homogeneous  susceptible class, but  vector ``diseased" state $\vi $ (which may contain latent/exposed, infective , asymptomatic, etc) and vector removed states (healthy, dead, vaccinated, etc). It is defined by an ODE  system:
\be{syr}
\begin{aligned}
\bc
& \s'(t)= - \s(t)\; \vi  (t)  \bb \\
&  \vi '(t)=    \s(t)  \; \vi (t)    \bb \; \va +\vi (t) A \\
&\vec \r'(t) =  \vi (t) W
\ec
\end{aligned}
\ee
where %the transmission rate $\b$ is a constant and
\BEN
\im $\vi (t) \in \mathbb{R}^n$ is a row vector whose components are fractions of diseased individuals of various types,
which must \saty $\vi (\I)=0$.
\im $ \bb \in \mathbb{R}^n$ is a column vector whose components represent the relative transmission ability of the various disease classes.  %The sum  of its components will be denoted by $\b$.

\im $\va \in \mathbb{R}^n $ is a {\bf probability row vector} with the components representing the fractions of susceptibles entering into the corresponding disease compartments, when  infection occurs.
\im  $A$ is a $n\times n$ Markovian sub-generator matrix describing  rates of transition between the diseased classes $\vi $ (i.e., a Markovian generator matrix for which the sum of at least one row is strictly negative). Alternatively, $V:=-A$ is a  non-singular M-matrix.\fn[4]{An M-matrix is  a real matrix $V$ with  $ v_{ij} \leq 0, \forall i \neq j,$ and having eigenvalues whose real parts are nonnegative \cite{plemmons1977m}.}
\im $\vec \r(t) \in \mathbb{R}^p $ is a row vector which must \saty $\vr(\I)>0,$ whose components represent (fractions of) various  classes
which survive at the end of an infection.
\im $W\in \mathbb{R}^n \times \mathbb{R}^p$ is a matrix whose components represent the rates at which classes of diseased individuals become recovered. We assume   that the matrix $\T V =\bep A &W \eep \in \mathbb{R}^{m+n} \times \mathbb{R}^n$ has row sums $0$, which implies mass conservation.

\EEN
\eeD

We turn now to a   deceivingly simple particular example  of the SIR-PH model, which explains its name.
 \beR {\bf Probabilistic interpretation of SIR-PH epidemics}. For simplicity, let us group  all the output classes of \eqref{syr} into one $\r=\vr  \bff 1$,  yielding:
 \be{R21} \bc \s'(t)=- \s(t) \vec \i(t) \bff \b\\
\vec \i'(t)=\s(t) \vec \i(t) \bff \b \va + \vec \i(t) A\\ \r'(t)=\vec \i(t) \bff a,\ec \ee
where we put $\bff a:=W \bff 1=(-A) \bff 1$.

\eqr{R21} emphasizes the fact that SIR-PH models are in one to one correspondence with laws of \PH\ $(\va, A)$ \cite[(21)]{Riano}.

Let us recall now, as known essentially since \cite{Kurtz} -- see also \cite[Thm. 2.2.7]{Britton}--
that under proper scaling, the  expected {fractions} $\s(t)$, $\i(t), \r(t)$  of stochastic SIR\fn[4]{One such model  stipulates that each infective $j=1,...,J$ infects a randomly chosen susceptible, at encounter times which belong to independent    Poisson processes $P^j(t), j=1,...,J$, of rate $\b$, and that   infection durations  are i.i.d. r.v.'s  which are
{\bf exponentially distributed} %with parameter $\g$,
at the end of which the individual recovers (or dies).} and more general { compartmental models} obey  a ``law of large numbers/fluid limit"  which recovers the \det\ epidemic.

As an example, the SIR-PH model \eqref{syr}
 may be derived as limit  of a stochastic SIR model in which the exponential infection time  has been replaced by
 a \PH\ $(\va, A)$ ``dwell period" \cite{Hurtado}.
\eeR
 \beP \la{p:Arino}
For processes defined by \eqr{syr}, with $V=-A$ a non-singular M-matrix,
the basic reproduction number is given by \cite[Thm. 2.1]{Arino}\fn[4]{This can be also derived as the expected number of susceptibles infected during a dwell period, for the stochastic model (the so-called ``survival method")--see  \cite{perasso2018introduction} for an excellent review}.
\be{R} \mR=   \va  \;  V^{-1} \bb.\ee

A disease free equilibrium $(s_0,\vec 0, \vec r_0)$ is asymptotically stable iff $s_0 < \fr 1{\mR}.$

\eeP

To illustrate the power of the SIR-PH formalism, consider now the  case with two diseased states,   latent and  infectious,     with \PH\ dwell times, parametrized by $(\va_e,A_e)$ and $(\va_i,A_i)$, \resp.
Using the \wk\ convolution formula -- see \fe \cite[Thm. 3.1.26]{bladt2017matrix} we find
that  formulas like \eqr{R} (see other examples of such formulas below) still apply, with   $(\va, A, \bb)$
given by
 \be{SEIRPH} \va=(\va_e,0), A= \bep A_e& \ba_e \va_i \\ 0 &A_i \eep,
 \bff \b=\bep 0\\0\\ \vdots\\ \b_{i,1}\\ \b_{i,2} \\ \vdots\eep. %W=(0,\va_i),
\ee

The ``epidemic dwell strucure" $(\va, A, \bb)$  of examples with more complicated network structures for the diseased may be constructed using
Kronecker sums of the matrices defining each component.

Let us give now an example which {\bf does not in general} belong to the SIR-PH class.

\beXa \iffalse
The SIR model has been  useful for modeling qualitatively the current COVID-19 pandemic. In this context, it is useful to further divide the removed compartment into recovered (R)  and dead  (D).
Under some  objectives
this may not be necessary,  since the final size of the dead is  simply proportional to the final size of the removed; however, this model allows including a   vaccination controlled parameter $\v$ and  recovering individuals becoming  again susceptible, at a  rate $\rho$.
\fi
The  SIRV model (SIR with vaccination --see for example \cite{BBG}) is defined by:

\begin{align}
\s'(t)&=  -\s(t) \pr{\b \i(t)+ \g_s} ,\nonumber\\
\i'(t)&= \i(t)\pp{\b \s(t)-\g_i },\label{sirV}
\\
\r'(t)&=   \g_i \i(t),\nonumber\\
\v'(t)&=\g_s \s (t). \nonumber
\end{align}
\iffalse

\figu{sirv}{Plot of the states of \eqr{sirV}, for $\b=.4,\; \g_s=.05,\;\g_i=.06,\; s_0=.99,\; i_0=1-s_0,\; r_0=v_0=0 $
   }{.8}%SIRV
   \fi

This is of the form \eqr{syr} with $\vi =(\i), \vr=(\r,\v)$ iff $\g_s=0$.

In the opposite case $\g_s \neq 0$, one may still compute an invariant
\be{invv}\mu(s,i):=\b( s+i) -\g \ln(s)+ \v \ln(i), \ee
and for fixed $s$, putting $\T \g=\fr \g {\g_s}, \T \b=\fr \b {\g_s}, \T \mu_0=\fr {\mu_0} {\g_s}$, \ith
 $i$ is explicitly given by
\be{iR} i(s) =\fr 1{\T \b}L_0 \pp{ \T \b s^{\T \g} e^{\T \mu_0-\T \b s}}.\ee

 \im When $\g_s > 0$, the final size is $s_\I=0$.

\eeXa

{We provide now a list of several  formulas,  obtained by replacing  $\i$ in SIR  by a  scalar
linear combination \eqr{Y} \cite{Feng}. They are all easily proved; \how\
 the formula  for the
maximal value of the newly infected involves also a second linear combination $\y $ \eqr{y}. }

\beP \la{p:Feng}
For processes defined by \eqr{syr}, with $V=-A$ a non-singular M-matrix, \ith:
\BEN

\im The following  \textit{weighted sum of the diseased variables}   \cite[(24)]{Feng}
\be{Y}
Y(t)= \fr{\vi (t)  \;  V^{-1} \bb}{\va  \;  V^{-1} \bb} =\fr{1}{\mR} \; \vi (t)  \;  V^{-1} \bb
\ee
has the property that
\be{invd} \frac{dY}{d\s} =  \fr {\fr{1}{\mR} \vi (t) \pr{  \s(t) \;   \bb \; \va - V}V^{-1} \bb}{- \s(t)\; \vi  (t)  \bb}=
\fr { \vi (t) \bb \; \pr{  \s(t) \;    \mR -1}}{-{\mR} \s(t)\; \vi  (t)  \bb}=
 -1 + \frac{1}{\mR \s}, \ee
 \text{ and that }
\be{inv-SYR}\bc  z(t)=\mu(\s(t),Y(t)):=Y(t)+ \s(t)-\fr 1 {\mR} {ln[\s(t)]},\\
 Z(t)=e^{ -{\mR} z(t)} =\s(t) e^{-\mR (\s(t)+Y(t))} \ec\ee
are constant along the paths of the dynamical system \eqr{syr}.

The solution of $Z(s)=Z(0)$  \wrt\ $\s$ \mbe
\be{sY}
  \s(t) =-\fr {1}{\mR}  L_0 \pp{-{ \mR}{Z_0}  e^{\mR Y(t)}},\ee
   where $[-e^{-1}, \I) \ni z \to L_0(z) \in [-1,\I) $ is the principal branch of the  Lambert-W function.

   \im The derivative with respect to time is
\be{y}\frd{Y}{t}=\left(\s-\fr 1{\mR}\right)\vi  \bb:=\left(\s-\fr 1{\mR}\right)\y.\ee
\Thr $\frd{Y}{t}=0=\fr{dY}{ds}$ iff $ s=\mR^{-1}$.

\im
The maximum value of $Y$ occurs for $s =  \min[1/\mR,1]$.
In the case $\mR >1$,
this yields \cite[Sec. 2.1]{Feng}:
\be{mp} Y_0 +\s_0 - Y^*-\fr 1 {\mR} = \fr 1 {\mR} \ln({\s_0 \mR}),\ee
by the conservation of $Y(t)+\s(t)-\fr 1 {\mR} \ln({\s(t)})$ between  the time $0$ and the time $t_{1/\mR}$ of reaching the immunity threshold).

\im The final size  of the susceptibles  \sats \cite[Thm.5.1]{Arino}:
\be{fs}  {ln[\s_0/\s_\I]}={\mR}\pr{\s_0 -\s_\I} +  \vi _0 V^{-1} \bb={\mR}\pr{\s_0 -\s_\I+ Y_0}, \ee
by the conservation of $Y(t)+\s(t)-\fr 1 {\mR} \ln({\s(t)})$ between  the times $0$ and $\I$; explicitly,

\be{si}
  \s_\I =-\fr {1}{\mR}  L_0 \pp{-{ \mR}{Z_0}}=-\fr {1}{\mR}  L_0 \pp{-{ \mR}{s_0 e^{- \mR(s_0 + Y_0)}}}\ee

\im   The integrated infectives  $\vI{a,b}=\int_a^b \vi (s) ds$
\sats\
\be{ini} \bc \vI{a,b} \;  V= \vJ_a-\vJ_b, \vec J_s:= \vi (s)+ s \va\\
 \pr{\vJ_a-\vJ_b} V^{-1} \bb= \log(\fr{\s(a)}{\s(b)}) \ec, \ee
 and the total integrated infectives  $\vI{\I}=\int_0^\I \vi (s) ds$
\sats\ \cite[(6)]{Arino}
\be{ini-total} \vI{\I} \;  V= \vi _0+ (s_0-s_\I)\va . \ee

\beR \la{r:int} In particular, for the SIR model \eqr{sir},
\be{Isi}\log\left(\fr{\s(a)}{\s(b)}\right)=\b I^{(a,b)} =\mR (J_a-J_b), \; J=\s+ \i.\ee
%\itf $I^{(a,b)} $ determines both the final values   $\s(b), \i(b)$.

Note that this has been used to model the total cost of an epidemic \cite{gani1972cost}.
\eeR

\im The final size  of the removed  \sats:
\be{rec} \vr_\I-\vec r_0 =I^{(\I)} \;  W=\pr{\vi _0+ (s_0-s_\I)\va} V^{-1} W , \ee

\im

The  value of the infected combination  $Y$  when $\s=1/\mR$  is
 \be{im-SYR}Y_{\max}=i_{\max,\mR}(s_0,\vi _0)= Y_0+s_0 -  \fr {1+ \ln(s_0 \mR)} {\mR}  =i_0 + \f{ H(s_0 \mR)}{\mR}, \; H(\mR)= {\mR-1-\ln(\mR)}.
\ee

\im The maximum size of the newly infected is achieved when
    \be{sdsmax} s(t) =\frac{\y^2 +\mR Y(t)}{\va \bb \y}. %= \fr{1}{\va \bb} \pr{y(t) + \mR \fr{Y(t)}{y(t)}}
    \ee
%\im \red{The maximum size of the newly infected is achieved when $\vi $ is an eigenvector of $\s \bb \va -V$, corresponding to the dominant eigenvalue which equals $\y$ ?}
\EEN
\eeP

\beR Let us note that for control problems involving optimization objectives which only depend on $Y(t)$, we are  effectively optimizing a SIR model;
this SIR approximation may be used to offer practical solutions for optimizing more complicated compartmental models. %Also,  the original parameters of the model  enter the invariance equations only via the formula of $\mR$.
\eeR

\beXa For SEIR, putting $\vi  =(\e , \i)$, we may write
 \bea %{seirex}
 \bc
 \s'= -\b \s \i \\
 \vi  '=   (\b \s \i  - \g_e \e, \g_e \e  -\g \i)= \b \s \i (1,0)  -(\e,\i) \bep  \g_e & -\g_e \\ 0& \g \eep   \\
 \r'= \g \i
 \ec \eea
so that   $\bc  \va= (1,0) \\  \bff b= \bep
0 \\\b
\eep\\V= \bep  \g_e & -\g_e \\ 0& \g \eep \Lra V^{-1} = \fr 1{\g \g_e} \bep  \g & \g_e \\ 0& \g_e \eep,  Y = \fr{\bep \e,\i \eep \bep  \g & \g_e \\ 0& \g_e \eep \bep
0 \\\b
\eep}{\bep 1,0 \eep \bep  \g & \g_e \\ 0& \g_e \eep \bep
0 \\\b
\eep}=\e +  \i\\ W= \bep 0&\g\eep \ec.$

\iffalse
We add here a point raised by \cite{Post}: ``
the stage of illness each infected person enters the statistics as
a registered infected one is less reflected in the available data. Hence, the division into latent (exposed) and active forms, i.e any
kind of SEIR-models might be over-complexification with respect to
the actual data uncertainty."
\fi

\eeXa

\section{Examples of SIR-PH models used in COVID-19 modelling \la{s:exa}}

{We derive now $\mR$ and $Y$ from \eqr{R},  \eqr{Y}, for some  popular compartmental models.} Note that we will be  reformulating the original results (which, unfortunately,  have already appeared several times with different notations), using  a
unifying notation.

\beXa The SEIHRD model  \cite{ivorra2017stability,Palmer,pazos2020control,nave2020theta,ramos2021simple} has disease states $\vi  =(\e , \i, \h)$. We use here  the version in  \cite{pazos2020control} (we would rather call this S$I^2$HRD model),   defined by $$\va=\bep 1,0,0\eep, \bb=\bep \b_e\\\b_i\\0\eep, V=\left(\begin{array}{ccc}
 \g_e  & -e_i  & 0 \\
 0 & \gamma_i  & -i_h \\
 0 & 0 & \g_h \\
\end{array}
\right), W =\bep e_r&0\\i_r&0\\h_r&h_d \eep,
$$ where we denoted by $\g_e,\g_i, \g_h$ the sum of the constant rates out of $\e, \i, \h$, and by $i _h$ the rate out of $\i$ and reaching $\h$, etc. Then, $\mR=\fr {\b_e}{\g_e}+ \fr {e_i }{\g_e  } \fr { \b_i}{\g_i }$ \cite[2]{pazos2020control}, %\cite[Thm 1]{ivorra2017stability},
and $ Y=\e+\i \frac{ \gamma _e \beta _i}{\beta _e \gamma _i+e_i \beta _i}.$  When $\b_e=0=e_r \Lra \fr{e_i}{\g_e} =1$, we recover $\mR=\fr {\b_i}{\g_i}$ \cite{Palmer} and $ Y=\e+\i.$

\begin{figure}[H]
\centering
\includegraphics[scale=0.8]{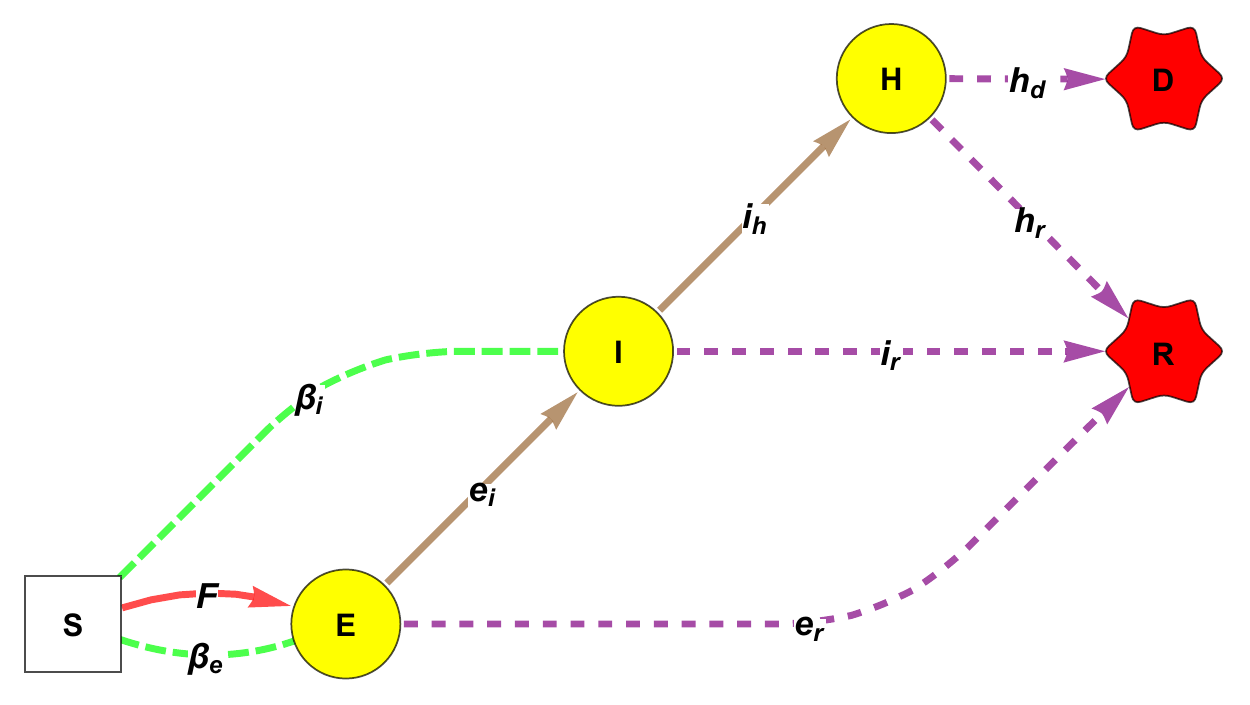}
\caption{Chart flow of the SEIHRD model. The forces of infection  are $F_{se}= \b_e \e,\; F_{si}= \b_i \i  , F=F_{se}+F_{si}.$
The red edge corresponds to the entrance of susceptibles into the diseased classes; $\va$, the dashed green edges correspond to contacts between  diseased to susceptibles, the brown edges are the rates of the transition matrix $V$, and the remaining yellow dashed flows correspond to the rates of $W$. \label{f:SEIHRD}}
\end{figure}
\eeXa
\beXa The SEIHCRD model of \cite{Kantner} has disease states $\vi  =(\e , \i, \h, \c)$ and is defined by $$\bb=\bep 0\\\b_i\\0\\0\eep, \va=\bep 1,0,0,0\eep, V=
\bep
 \g_e  & -e_i &0  & 0 \\
 0 & \g_i  &-i_h& 0 \\
 0 & 0& \g_h &-h_c \\
 0 & 0 &-c_h& \g_c \\
\eep , \; W =\bep 0&0\\i_r&0\\h_r&0 \\
0 & c_d\eep%\Lra
,$$
then,
$$ \mR=\frac{\b_i}{ \gamma _i}, Y=\e+\i.$$

 \begin{figure}[H]
\centering
\includegraphics[scale=0.8]{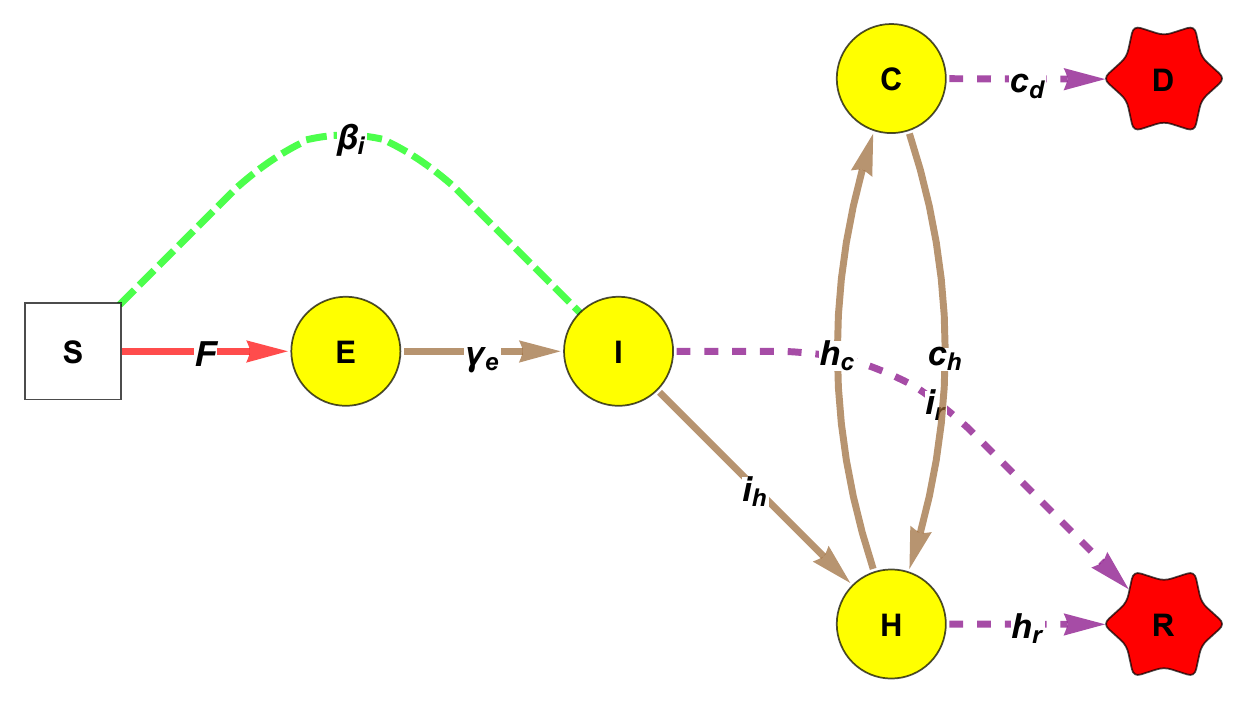}
\caption{Chart flow of SEIHCRD model. %The red edge corresponds to the entrance of susceptibles into infectious class; $\va$, the dashed green edge corresponds to the flux of contact between infected and susceptibles , the brown edges are the rates of the transition matrix $V$, and the remaining black flows correspond to the rate of $W$.
The force of infection  is $ F=\b_i \i.$
\label{f:SEIHCRD}}
\end{figure}
\eeXa

\beXa The SEIRAH(SEIAHR) model  \cite{%coelho2020modeling,
deng2021extended,prague2020population} has disease states $\vi  =(\e , \i, \a,\h)$ and
is defined by $$\bb=\bep 0\\\b_i\\\b_a\\0\eep, \va=\bep 1,0,0,0\eep, V=
\bep
 \g_e  & -e_i &-e_a  & 0 \\
 0 & \g_i  &0& -i_h \\
 0 & 0& \g_a &-a_h \\
 0 & 0 &0& \g_h \\
\eep,\;  W =\bep e_r\\i_r\\ a_r \\ \g_h \eep.$$

Then, $$\mR=\frac{ e_a}{\gamma _e } \mR _a +\frac{e_i }{\gamma _e } \mR _i, \quad {\mR _i }=\frac{\b _i }{\g_i}, {\mR _a }=\frac{\b _a }{\g_a},$$ %\cite{otoo2021estimating}
and $$ Y=\e+\i \frac{\mR _i }{\mR}+\a \frac{\mR _a }{\mR}.$$
 %As a check, note that when $\a=0,\b_i=\b_a$ (in the absence of asymptomatics), these formulas reduce to the ones in the SEIHRD example.
   \begin{figure}[H]
\centering
\includegraphics[scale=0.8]{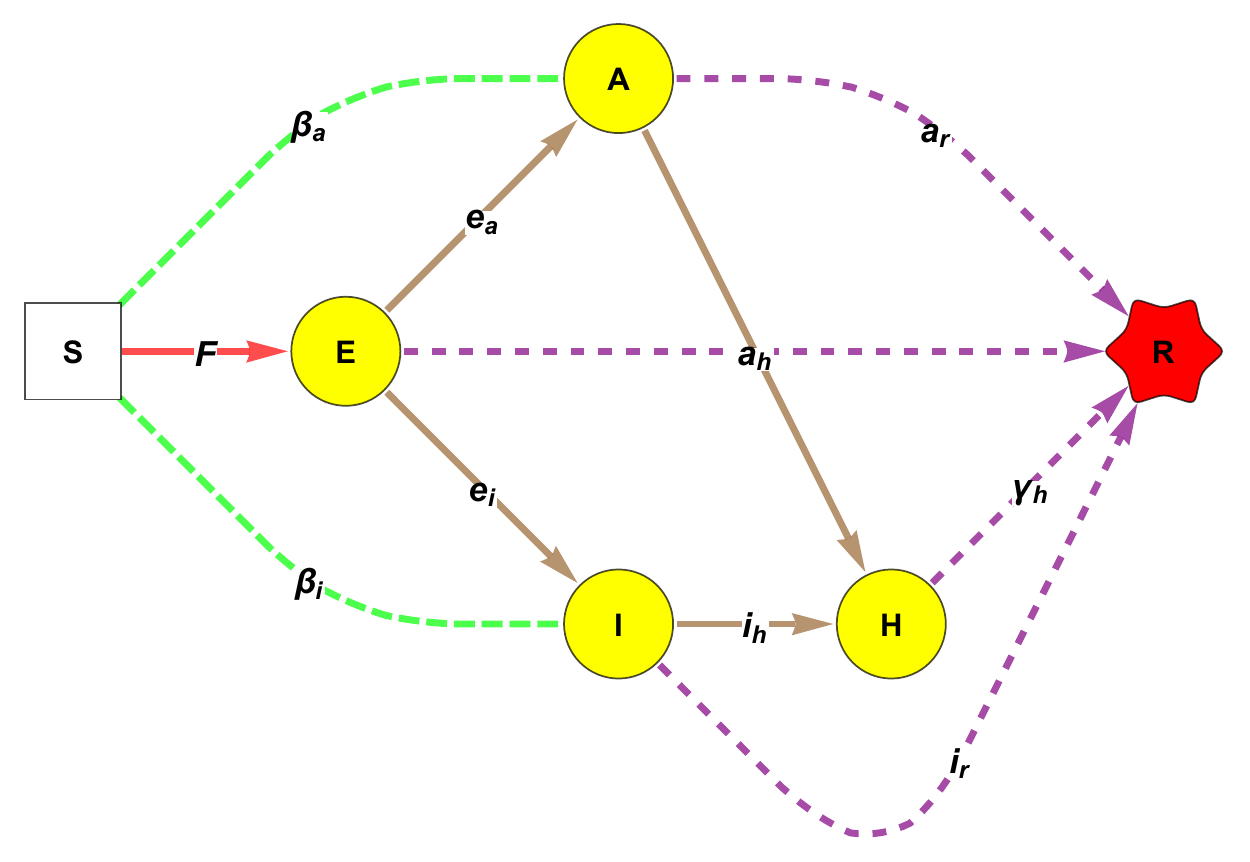}
\caption{Chart flow of the SEIAHR model. %The red edge corresponds to the entrance of susceptibles into infectious class; $\va$, the dashed green edges correspond to the flows of contact between infected and susceptibles, the brown edges are the rates of the transition matrix $V$, and the remaining black flows correspond to the rate of $W$.
The forces of infection  are $F_{si}= \b_i \i,\; F_{sa}= \b_a \a  , F=F_{si}+F_{sa}.$
\label{f:SEIRAH}}
\end{figure} \eeXa

\beXa The S$I^{aps}$QR model (with asymptomatic, pre-symptomatic and  symptomatic infectious)  \cite[(3.2)]{shaw2021reproductive}, \cite{Hadany}.

\begin{figure}[H]
\centering
\includegraphics[scale=0.8]{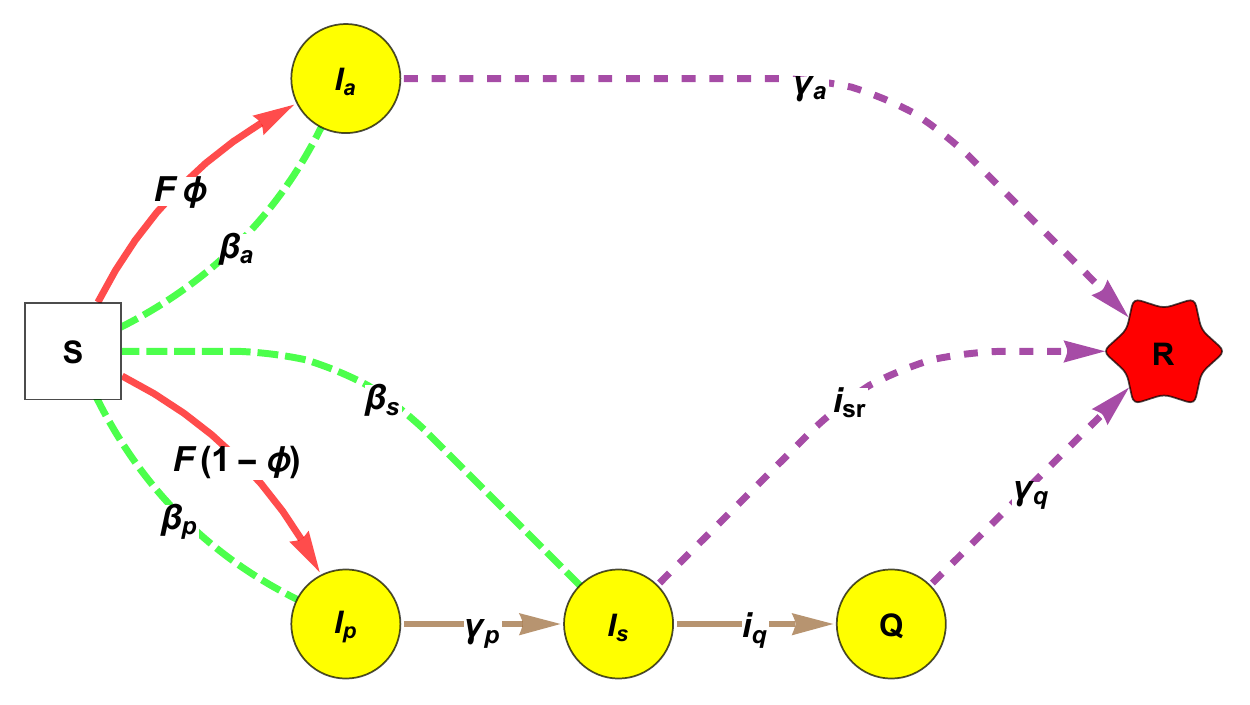}
\caption{Chart flow of S$I^{aps}$QR model.  The forces of infection (rates) are $F_{sa}= \b_a \i_a,\; F_{sp}= \b_p \i_p,\; F_{ss}= \b_s \i_s  , F=F_{sa}+F_{sp} + F_{ss}.$\label{f:SIapsQR}}
\end{figure}

The disease states are $\vi  =(\i_a , \i_p, \i_s,\q)$, and the model is defined by
 $$\bb=\bep \b_a\\\b_p\\\b_s\\0\eep, \va=\bep \phi,1-\phi,0,0\eep, V=
\left(
\begin{array}{cccc}
 \gamma _a & 0 & 0 & 0 \\
 0 & \gamma _p & -\g_p & 0 \\
 0 & 0 & \gamma _s & -i_q \\
 0 & 0 & 0 & \gamma _q \\
\end{array}
\right),\; W =\bep \g_{a}\\ 0\\ i_{sr} \\ \g_{q} \eep.$$

Then, \cite{shaw2021reproductive}, $$\mR= \phi \frac{\b_a}{\g_a}+ (1-\phi)\left( \frac{\b_p}{\g_p}+ \; \frac{\b_s}{\g_s}\right),$$

and, $$Y= \frac{1}{\mR} \; \left[\i_a \; \mR_a + \i_p(\mR_p +\mR_s) + \i_s \mR_s  \right], $$
 where $\mR_a:= \frac{\b_a}{\g_a}, \; \mR_p:= \frac{\b_p}{\g_p}, \; \mbox{and} \quad \mR_s:= \frac{\b_s}{\g_s}.$

\eeXa

\section{$S^{(m)}YR$ models with $m$ groups of susceptibles \la{s:het}}
The SIR-compartment model  makes the unrealistic assumption that the population through which the disease is spreading is well-mixed.
However, differences in susceptibility and rates of contact between individuals strongly affect their likelihood of catching COVID-19. A model which attempts to capture this  aspect is:
\be{syrW}
\begin{aligned}
\bc
& \s_k'(t)= -  \s_k(t)\; \vi (t)  \bb_k, \; \; k=1,...,m \\
&  \vi'(t)=   \sum_{k=1}^m   \s_k(t) \; \vi(t)  \bb_k \;  \va -\vi(t) V \\
&\vec \r'(t) =  \vi(t) W
\ec
\end{aligned}
\ee

\beL
A disease free equilibrium $(s_1, s_2, \dots, s_m, \vec 0, \vec r_0)$ of \eqref{syrW}
is asymptotically stable iff $s \mR <1,$ where $s = \sum_k s_k$ and
 \be{R0mat}
 \mR =  \sum_k \frac{s_k}{s} \;  \va\; V^{-1} \bb_k=  \sum_k \frac{s_k}{s} \;  \mR_k, \quad \mR_k:=\va\; V^{-1} \bb_k\ee
 is the spectral radius of the next generation matrix.
\eeL

While the final size may also be obtained under this model \cite[Thm. 2.1]{Andr}, for transient behavior
it is convenient  to turn to a simpler model.

\ssec{A generalization of heterogeneous SEIR \cite{Dolb}}

  Assume now that $\bb_k =\b_k \bb,$ where $\b_k \in \R_+$ and W=$\bb \vec w$,
  where $\vec w$ is a row vector. Putting  $\y=\vi \bb$,
the  system \eqr{syrW} becomes:\fn[4]{Such a dynamics was first considered in \cite{Gart}.}
\begin{equation}
\frac{d \log\s_k}{dt} = - \beta_k\,\y\,,\quad
\frac{d\vi}{dt}= \(\,\sum_{k=1}^m\beta_k\,\s_k\)\y \va -  \vi\,V,\quad
%\frac{d\i}{dt}=\g_e\,\e-\,\gamma\,\y\,,\quad
\frac{d\vr}{dt}=\vi W=\y \vec w,
\label{syr2}
\end{equation}
and \be{I}{\vr(t)}=\int_0^t \y(u)du \; \vec w:=I(t) \; \vec w.\ee

It is convenient to reparametrize the model taking $I$ as parameter, or, equivalently, by taking
   \be{gam}\r =\vr \bff 1=\g I, \g:= \vec w \bff 1.\ee

   Solving $$\frac d{dt}\(\log\s_k\)=-\,\beta_k\,\y=-\,\frac{\beta_k}\gamma\,\frac d{dt}\r\, ,$$
 we find that the system has a family of first integrals which includes
\be{seir2:conservation}
 \bc \fr{1}{\mR_1} \log(s_1/s_1(0))=...= \fr{1}{\mR_k} \log(s_k/s_k(0))=...= \fr{1}{\mR_m} \log(s_m/s_m(0))=r(0) -r,\\
\sum_{k=1}^m s_k+\sum_{k=1}^n y_k + \sum_{k=1}^p r_k=1\ec.
\ee

Also
\be{s_k(t)}
\s_k(t)=\s_k(0)\,e^{-{\beta_k\,I(t)}}=\s_k(0)\,e^{-{\fr{\beta_k}\g\,\r(t)}},
%\quad\mbox{with}
%\quad\s_k(0)=\s_k(0)\,e^\frac{\beta_k\,\r(0)}\gamma,\;
%\s_k(0)=p_k\,\s(0).
\ee
where $\g$ is defined in \eqr{gam}.

 We conclude with some preliminary  results on this model.

%-------------------------------------------------------------------------
\begin{lemma} \la{l:S}  a) \cite{Dolb} Put $\s(t)=\sum_k \s_k(t), p_k = \s_k(0) / \s(0)$.%, and let $B$ be the matrix having $b_k$ as columns.
The solution  of~\eqref{syr2} satisfies
the time dependent SYR system:
\begin{equation}
\frac{d\s}{dt} = -\,\,a(t)\,\s(t)\,\y(t) \,,\quad\frac{d\vi}{dt} = \pr{\sum_k \beta_k\, \s_k(t)}\,\y(t) \, \va - \vi(t) V\,,\quad\frac{d\r}{dt} = \gamma\,\y(t) \,,
\label{eq:alternativeseir}
\end{equation}
where
\be{def:a}
a(t)=\frac{\sum_k \beta_k\, \s_k(t)}{\s(t)}
\ee
 is a positive non-increasing function  with $a(0)= \sum_k p_k\,\beta_k:=\bar{\beta}  $.
\label{lemma:attentuation}

b)  $Y(t)$  defined in \eqr{Y} with $\mR=\sum_k p_k \mR_k, \mR_k  = \b_k \va V^{-1} \bb,$ \sats:
\be{Yt}\fr 1{\y(t)} \frd{Y}{t}= \fr{1}{\mR}(\sum_k  s_k(t) \mR_k-1)=\fr{\mR_e(t)-1}{\mR}=\g  \frd{Y}{\r}\ee
and is unimodal, with a maximum on the immunity/recovery line
\be{Imline} \sum_k  s_k {\mR_k}= 1.\ee

c)   The stable stationary solution $(\s_1^\star,\s_2^\star,...,0,0,...,\vr^\star)$ is determined  by the unique solution with $\r^\star=r>\r(0)$ of
\be{Eq:Equilibrium2}
1=\sum_{k=1}^m\s_k(0)\,e^{-\frac{\beta_k\,r}\gamma}+r\,,\quad\s_k^\star=
\s_k(0)\,e^{-\frac{\beta_k\,r}\gamma}\,,\quad k=1,\,2, \,...\, ,m.
\ee

\end{lemma}
%-------------------------------------------------------------------------
\begin{proof}
a)  The derivative of $a$ \sats
\be{decay:a}
\frac{da}{dt} = - \frac{\vi(t) \bb}{\bar\beta\,\s(t)^2}\({\textstyle\sum_k\s_k(t)\,\sum_k \beta_k^2\,\s_k(t)-\big( \sum_k \beta_k\,\s_k(t) \big)^2}\)\le0
\ee
by the Cauchy-Schwarz inequality.

b) $Y(t)=\fr{1}{\mR}  \vi(t)  \;  V^{-1} \bb \Lra $
\bea &&Y'(t)=\fr{1}{\mR}  \vi'(t)  \;  V^{-1} \bb =
\fr{1}{\mR}
\pp{\pr{\sum_k\, \s_k(t) \, \beta_k} \y(t) \, \va - \vi(t) V}
 V^{-1} \bb\\&&=\fr{\y(t)}{\mR} \pp{\pr{\sum_k\, \s_k(t) \, \mR_k } -  1}
\eea

c) This follows from the conservation of mass and $\vi(\I)=0$ \cite{Dolb}.

\end{proof}

%\input{Appen}
%\input{D}
%\input{FS2}
%\input{pFe}
%\input{R0Cae} %inhom
\iffalse
\input{cosyr}
%\input{ap}

\input{py1}
\input{py2}
%\input{part1}
%\input{tableimR}

%\input{Palm}%
%\input{PalmR}
\fi

%

%\input{all0}
%\input{all}
%\input{all20}

\bibliographystyle{amsalpha}

\bibliography{Pare38}
%\bibliographystyle{abbrv}
%\bibliography{../../../freddi}

\end{document}